\documentclass[11pt,a4paper]{article}
\usepackage{enumerate,paralist}
\usepackage{amsmath,amsthm,amssymb,hyperref}
\usepackage{framed}
\usepackage{fullpage}
\usepackage{float}

\newcommand{\ceil}[1]{\lceil #1 \rceil}
\newcommand{\floor}[1]{\lfloor #1 \rfloor}
\newcommand{\ang}[1]{\langle #1 \rangle}
\newcommand{\set}[1]{\{ #1 \}}

\renewcommand{\epsilon}{\varepsilon}

\makeatletter
\newcommand\addleadingzero{\expandafter\ifnum
  \csname theenum\romannumeral\the\@enumdepth\endcsname<10 0\fi}
\makeatother

\usepackage{setspace}
\newenvironment{pseudocode}[2]{%
\vspace{2mm}\begin{singlespace}
\noindent \textbf{Procedure} {\it #1}(#2)
\begin{compactenum}[\addleadingzero1.]}
{\end{compactenum}
\end{singlespace}\vspace{2mm}}
\newcommand{\comment}[1]{\hspace{2mm}\textbf{//}#1\textbf{//}}
\newcommand{\tab}{\hspace{4mm}}
\newcommand{\stab}{\hspace{1.25mm}}

\title{\vspace{-18mm} \Large  A New Push-Relabel Algorithm for Sparse Networks}
\author{ Rahul Mehta\thanks{This work was conducted while the author was a student at the University of Chicago Lab High School, as well as a summer intern at the Toyota Technological Institute of Chicago. The author is currently a student at Princeton University.} \\ University of Chicago Lab High School  \\ {\tt rahulmehta@uchicago.edu}}
\date{\today}

\begin{document}
\maketitle

\begin{abstract}
In this paper, we present a new push-relabel algorithm for the maximum flow problem on flow networks with 
$n$ vertices and $m$ arcs. Our algorithm computes a maximum flow in $O(mn)$ time on sparse networks where 
$m = O(n)$. To our knowledge, this is the first $O(mn)$ time push-relabel algorithm for the $m = O(n)$ edge 
case; previously, it was known that push-relabel implementations could find a max-flow in $O(mn)$ time when 
$m = \Omega(n^{1+\epsilon})$ (King, et. al., SODA `92). This also matches a recent general flow decomposition-based algorithm 
due to Orlin (STOC `13), which finds a max-flow in $O(mn)$ time on near-sparse networks.

Our main result is improving on the Excess-Scaling algorithm (Ahuja \& Orlin, 1989) by reducing the number 
of nonsaturating pushes to $O(mn)$ across all scaling phases. This is reached by combining Ahuja and 
Orlin's algorithm with Orlin's compact flow networks. A contribution of this paper is demonstrating that the 
compact networks technique can be extended to the push-relabel family of algorithms. We also provide 
evidence that this approach could be a promising avenue towards an $O(mn)$-time algorithm for all edge densities.
\end{abstract}

\section{Introduction}\label{sec:intro}

The maximum flow problem has been studied for several decades in computer science and operations research. One of the most widely researched combinatorial optimization problems of all time, 
many polynomial-time algorithms have been presented since the groundbreaking research of Ford and Fulkerson \cite{ford56}. The max-flow problem has numerous theoretical applications in both computer science 
and operations research, as well as practical applications in transportation, scheduling, and routing problems, and more recently in computer vision. A full discussion of the history of the max-flow problem and its applications can be found
in \cite{ahuja93}. Many efficient algorithms exist, including ones based on augmenting paths, blocking flows, and the push-relabel method.

Recently, several asymptotically fast algorithms have been presented. Together, the results of King, et. al. \cite{king92} and Orlin \cite{orlin13} show that the max-flow problem is solvable in $O(mn)$ time on general networks
with $n$ vertices and $m$ arcs. The former is a derandomized version of an algorithm due to Cheriyan, et. al. \cite{cher90}, which relies on a push/relabel/add-edge approach. This runs in $O(mn)$ time when 
$m = \Omega(n^{1+\epsilon})$, and uses the correspondence between nonsaturating pushes
and a certain combinatorial game to bound the running time. 
The latter is based on the augmenting paths algorithm of Ford and Fulkerson, but leverages a smaller representation of the flow network known as the ``compact
network'' to reduce the time to run a scaling phase, in a manner similar to Goldberg and Rao \cite{gold98}. This algorithm runs in $O(mn)$ time when $m = O(n^{16/15 - \epsilon})$, and in $O(n^2 /\log n)$ time on sparse networks. 
Our algorithm utilizes these 
compact networks to achieve our improved running time. There have also been recent developments utilizing electrical flows and fast algorithms for approximately solving Laplacian systems of equations; 
for more information the reader is referred to \cite{christ11,kelner14,madry13}.

\paragraph*{Our contribution}
We show how to solve the max-flow problem in $O(mn+m^{3/2}\log n)$ time for sparse networks (i.e. when $m = O(n)$). Clearly, since $m = O(n)$, the algorithm runs in $O(mn)$ time.
This extends the known range for which max-flow is solvable in $O(mn)$ time using push-relabel algorithms (King, et. al. \cite{king92} gave
an $O(mn)$-time algorithm when $m = \Omega(n^{1+\epsilon})$.)
This matches the general algorithm of Orlin \cite{orlin13}, which solves the problem on general networks in $O(mn + m^{31/16}\log^2 n)$, and $O(mn)$ time on near-sparse networks. Although we do not
achieve the $O(n^2 / \log n)$ result on sparse networks, we give the first strongly-polynomial algorithm using the push-relabel method for the sparse edge case.
We reach the improved running time by running an Ahuja-Orlin scaling phase on a smaller representation of the flow network, known as the compact network. Through this approach,
we are able to obtain a strongly-polynomial algorithm for bounded-degree networks; that is, when there is some $k \in \mathbb{Z}_{> 0}$ such that for all $u \in V$, $\deg[u] \leq k$. We call
this a $k$-bounded-degree network.
Moreover, using a well-known technique, we can transform a graph $G$ where $m = O(n)$ into
an $(\floor{m/n}+3)$-bounded-degree network; thus, our bounded-degree algorithm can be extended to the more general case of sparse networks.
This transformation is described in Appendix \ref{app:sparse}.

An extension of our results to general networks, which would involve resolving certain
issues relating to nonsaturating pushes, would imply an algorithm that runs in $O(mn)$ time when $m = O(n^{2-\epsilon})$. We believe that 
an algorithm which incorporates one or more of our techniques may be a promising avenue towards
developing an algorithm that works for all edge densities.

\paragraph*{Our approach}
Our algorithm uses a modified version of the push-relabel method \cite{gold88}. We maintain a valid distance labeling $d:V \rightarrow \mathbb{Z}_{\geq 0}$ to estimate the distance of each vertex to the sink, and
push excess flow from higher-labeled vertices to lower-labeled vertices. We relabel (increase the distance label) of vertices to allow more pushes in a series of {\it scaling phases}, similar to those of Ahuja and Orlin \cite{ahuja89}.
During the $i$th phase, characterized by the parameter $\Delta_i$, which provides an upper bound on the individual excess of any vertex, 
we completely discharge $\Delta$-{\it active} vertices (i.e. vertices with an excess $\Delta_i/2 < e(u) \leq \Delta_i$), while maintaining the invariant that no excess rises
above $\Delta_i$. At the end of a phase, we guarantee that each vertex has an excess no greater than $\Delta_{i+1} \leq \Delta_i/2$. This generic technique, the Excess-Scaling algorithm of Ahuja and Orlin \cite{ahuja89}, leads
to an algorithm that runs in $O(mn + n^2\log U)$ time, where $U$ denotes the largest arc capacity. A careful analysis shows that by applying the techniques described in the following paragraph, we reduce the number of
phases to $O(m^{1/2})$ (Lemma \ref{lem:poly_phase}).

We adapt a recent technique of Orlin \cite{orlin13} to reduce the amount of time needed to run a $\Delta$-scaling phase. At the beginning of each phase, we construct a {\it compact network} consisting of a particular subset
of the vertices and arcs of the graph, along with an additional set of {\it pseudoarcs}. These pseudoarcs represent directed paths in the residual network. We include all vertices adjacent to {\it approximately medium} capacity 
arcs (which we call $\Delta$-favorable and $\Delta$-large), as well as all $\Delta$-active vertices. Then we construct $\Delta$-abundant and $\Delta$-small {\it pseudoarcs} (See Section \ref{sec:abundant}), 
which represent directed paths consisting of possibly many high-capacity (abundant) arcs, and low-capacity (small) arcs as a {\it single} arc in the compact network. 
When a path is found in the residual network, a pseudoarc to the terminal vertex is created with the capacity of the {\it bottleneck} arc on the
path (that is, the arc of minimal capacity). Then, the capacity of each arc on the path is decreased by the bottleneck capacity, and the process is repeated (this is called ``capacity transfer'' in Orlin \cite{orlin13}). We perform a slightly modified
version of the push-relabel algorithm on the compact network with the above guarantees.

With a careful analysis of the number of phases that specific arcs and vertices are included in the compact network, we arrive at a main result of our paper, which is that there are $O(m)$ vertices across all scaling phases
in the compact network (Theorem \ref{thm:om_comp}). Using a similar analysis, we limit the number of active vertices at $O(m)$ as well. Using observations about the behavior and size of saturating and 
nonsaturating pushes, we use several well-known potential functions from \cite{ahuja89,gold88} to reach the desired bound of $O(mn)$ nonsaturating pushes across all scaling phases.

\paragraph*{Major Results}
This paper negotiates several key technical difficulties associated with performing push-relabel on compact networks. First, the residual network must be updated after each capacity transfer. This is possible through the use of the
dynamic trees data structure of Sleator and Tarjan \cite{sleat83}, and is detailed in Appendix \ref{app:tc}. Second, we must guarantee that pushes along pseudoarcs do not violate the capacity constraint for any arc in the
original residual network. This does not trivially transfer to the generic push-relabel
method, since pushes along internal arcs within pseudoarcs may not be ``admissible'' in the traditional sense (see Section \ref{sec:disc}). We relax the Goldberg-Taran admissibility criterion from $d(u) = d(v) + 1$ to simply $d(u) > d(v)$, intuitively maintaining
the notion that ``flow must go downhill,'' and relabel the network at the end of each phase, using a technique due to Goldberg, et. al. \cite{gold97,gold97-2}. We define the {\it validity condition} for any arbitrary distance label $\text{dist}:V \rightarrow \mathbb{Z}_{\geq 0}$ as,
for all arcs $(u,v)$, $\text{dist}(u) \leq \text{dist}(v) + 1$. Since distance labels can increase due to both {\it low-capacity} nonsaturating pushes (those that send $\delta \leq \Delta/2$) {\it and} relabel operations, we maintain two distance labels 
$d_h(\cdot)$ and $d_\ell(\cdot)$ to keep track of the total. For sake of clarity, we let $d(u) = d_h(u) + d_\ell(u)$ denote the 
{\it overall} distance of $u$. Relabeling the entire network in the manner described above ensures that the overall distance $d$ obeys the validity condition at the end of a $\Delta$-scaling phase. 
We also maintain the invariant that throughout the execution of a scaling phase, $d_h$ obeys the validity condition as well; this, in turn, guarantees the correctness of our algorithm
(Theorem \ref{thm:valid_label}). $d_\ell$ does not necessarily obey the validity condition. However, it is relabeled only under specific circumstances, and thus is used to show that there are 
$O(mn)$ low-capacity nonsaturating pushes across all scaling phases (Lemma \ref{lem:small_nonsat_om}).

What is left to show is that at the end of every phase, we fulfill our promise, namely, that every vertex $u \in V$ that began the phase with excess $e(u) > \Delta/2$ is completely discharged (i.e. $e(u) = 0$),
and the excesses of the remaining vertices have fallen
below $\Delta_i/2$ at the termination of the $i$th scaling phase. In order to prove this, we show that pseudoarc pushes that send flow from some $u$ to $v$ such that $d(u) > d(v)$ allow the potential function $\Phi_g =
\sum_{u:e(u) > 0} d(u)$ to behave as it would in the generic push-relabel algorithm if we permit pushes along directed paths (this is formalized in Section \ref{sec:valid_push_comp}). Finally, we need to show that our novel
scheme for counting low-capacity nonsaturating pushes (ones that send $\delta \leq \Delta/2$)  correctly discharges active vertices (Lemma \ref{lem:act_correct}).

The claimed running time of $O(mn + m^{3/2}\log n)$ (Theorem \ref{thm:strong_rtime}) follows from the following main bounds;  \\
\begin{enumerate}[(1)]
\item At the end of every $\Delta$-scaling phase, {\it global-relabel} generates a valid labeling in $O(m)$ time. Across the $O(m^{1/2})$ scaling phases, the overall cost is $O(m^{3/2})$ (Theorem \ref{thm:valid_label});

\item There are $O(mn)$ nonsaturating pushes across all $\Delta$-scaling phases (Theorem \ref{thm:nonsat_bound_all});

\item Both the constructing the compact network $G_C$ and transforming it back into the residual network $G_f$ takes at most $O(m \log n)$ per scaling phase. The cost across all $\Delta$-scaling phases is $O(m^{3/2}\log n)$
(Theorem \ref{thm:strong_rtime}).
\end{enumerate}


\section{Preliminaries}
\subsection{Definitions}
A {\it flow network} is a directed graph $G=(V,A)$, with $|V| = n$ vertices and $|A| = m$ arcs. There are two distinguished vertices, namely the {\it source} $s$ and the {\it sink} $t$. Further, each pair $(u,v) \in V \times V$
has a non-negative, integer-valued capacity $c: V \times V \rightarrow \mathbb{Z}_{\geq 0}$. If $(u,v) \not\in A$, then $c(u,v) = 0$. The largest arc capacity is denoted by $U$.

A {\it flow} is an integer-valued function $f: V \times V \rightarrow \mathbb{Z}_{\geq 0}$ satisfying the {\it capacity} and {\it conservation} constraints, namely that (1) for all $(u,v) \in A$, $f(u,v) \leq c(u,v)$, and (2) for each
$u \in V \backslash \set{s,t}$,
	$$\sum_{v \in V} f(v,u) - \sum_{v \in V} f(u,v) = 0.$$
That is, the amount of flow entering each vertex $u$ is the same as the amount exiting $u$.
We denote the {\it value} or {\it magnitude} of the flow as the amount leaving the source (or
equivalently, entering the sink). That is,
	$$|f| = \sum_{u \in V} f(s,u) = \sum_{u \in V} f(u,t).$$
We say that an arc $(u,v) \in A$ is {\it saturated} if $f(u,v) = c(u,v)$
	
The {\it maximum flow problem} is to find a flow $f$ of maximum value. That is, we must find a function $f:V \times V \rightarrow \mathbb{Z}_{\geq 0}$ that obeys both the capacity and conservation
constraints and maximizes the flow into the sink $t$.
	
Next, we define an important concept in designing algorithms for the max-flow problem. The {\it residual capacity} with respect to a flow $f$ is an integer-valued function $r: V \times V \rightarrow \mathbb{Z}_{\geq 0}$ that is defined as
	\begin{equation*}
		r(u,v) = \begin{cases} 
		c(u,v) - f(u,v) & \text{ if } (u,v) \in A \\
		f(u,v) & \text{ if } (v,u) \in A \\
		0 & \text{ otherwise.}
	\end{cases}
	\end{equation*}
Intuitively, this allows us to express how much {\it more} flow can be sent on any given arc. We can assume without loss of generality that for each arc $(u,v) \in A$, the flow is nonnegative
for at least $(u,v)$ or $(v,u)$. Let the {\it residual network} be the
flow network $G_f=(V,A_f)$ consisting of pairs $(u,v) \in V \times V$ with nonzero residual capacity.

Next, we define several terms relating to cuts in the flow network. An {\it $s-t$ path} is any simple, directed path $\ang{s,v_1,v_2,\ldots,v_k,t}$ in the residual network from the source $s$ to the sink $t$
An {\it $(S,T)$-cut} is a bipartition of the vertex set $V$, such that $s \in S$ and $t \in T$, and $T = \bar{S}$. The {\it capacity} of the cut, denoted by $c(S,T)$, is equal to the capacity
of the arcs crossing the cut from $S$ to $T$.

We now state two classical results of Ford and Fulkerson \cite{ford56} that will be used in the analysis of the algorithm. The following is a direct consequence of the Max-Flow Min-Cut Theorem;
\newtheorem{lem}{Lemma}[section]
\begin{lem}\label{thm:maxflow} A flow $f$ is a maximum flow if and only if there does not exist an $s-t$ path in the residual network $G_f$. \end{lem}

\begin{lem}[Flow Decomposition Theorem]\label{thm:decomp} Any $s-t$ flow $f$ can be decomposed into flows $f_1 \ldots f_k$ on paths $p_1,\ldots,p_k$ in the flow network $G$. Moreover,
given a feasible flow $f$ and a maximum flow $f_{max}$, the flow $f' = f_{max} -f$ can be decomposed into flows  $f'_1 \cdots f'_{k'}$ on paths $p'_1,\ldots,p'_{k'}$ in the network $G_f$.\end{lem}

\subsection{The Push-Relabel Algorithm and Scaling}\label{app:prelim}

We must first introduce the notion of a {\it preflow}, and then can proceed to outline the generic push-relabel method of Goldberg 
and Tarjan \cite{gold88}, as well as Ahuja and Orlin's approach.

A {\it preflow} is a function $f:V \times V \rightarrow \mathbb{Z}_{\geq 0}$ that obeys the {\it capacity} constraint, but relaxes the {\it conservation} constraint; that is,
	$$ e(u) = \sum_{v \in V} f(v,u) - \sum_{v \in V} f(u,v) \geq 0, \text{ for all } u \in V \backslash \set{s,t}. $$
Intuitively, vertices may {\it overflow}. We define this discrepancy as the {\it excess}, denoted by the function $e:V \rightarrow \mathbb{Z}_{\geq 0}$.

A {\it distance labeling} is a function $d:V \rightarrow \mathbb{Z}_{\geq 0}$ that associates each vertex with a positive integer. Goldberg and Tarjan define $d$ to be a {\it valid} label if $d(s) = n$, $d(t) = 0$, and for all
$(u,v) \in A_f$, $d(u) \leq d(v) + 1$. The distance label allows us to push flow ``downhill.'' More formally, flow is sent along {\it admissible} arcs; 
in Goldberg and Tarjan's approach these are the arcs  $(u,v) \in A_f$ where $d(u) = d(v) + 1$. 

We generalize this notion and redefine {\it admissible} as simply $d(u) > d(v)$. In Section \ref{sec:valid_push_comp}, we show that one can still prove the correctness of the push-relabel algorithm using this definition of admissible arc, and increase the running time by at most a multiplicative constant. In addition, we apply the Goldberg-Tarjan notion of validity only to $d_h$; that is, $d_h$ is valid if, for all arcs $(u,v)$, $d_h(u) \leq d_h(v) + 1$. We do not require such criterion for the labeling $d_\ell$, and regard it as a
``restricted" distance label, of sorts; we only increment $d_\ell$ under only very specific circumstances, as outlined in Section \ref{sec:label}. Below, we formalize what we mean when we refer to a label being {\it valid};

\newtheorem{cond}{Condition}
\begin{cond}[Validity Condition]
A distance label $\text{\upshape dist}:V \rightarrow \mathbb{Z}_{\geq 0}$ satisfies the validity condition if, for all arcs $(u,v)$, $\text{\upshape dist}(u) \leq \text{\upshape dist}(v)+1$.

\end{cond}

The {\it push-relabel} family of algorithms, introduced by Goldberg and Tarjan \cite{gold88} maintains a preflow $f$, and iteratively ``pushes'' flow from some overflowing vertex 
$u$ to $v$ along admissible arcs. A vertex is ``relabeled'' when it is not incident to any admissible arcs.

Ahuja and Orlin \cite{ahuja89} modified the push-relabel algorithm to use a capacity scaling approach. Their algorithm (as well as ours) redefines an {\it active} vertex
to be any $u \in V$ where $\Delta/2 < e(u) \leq \Delta$, where $\Delta$, the  {\it excess dominator}, is an upper bound on the excess. By sending flow from $u = \min\set{d(w) | \Delta/2 < e(w) \leq \Delta}$
(the active vertex with the smallest label), we can guarantee that flow is sent to a vertex $v$ such that $e(v) \leq \Delta/2$ (Lemma \ref{lem:disc1_cond})

The Ahuja-Orlin algorithm performs a series of $O(\log U)$ {\it scaling phases} to find a max-flow. It maintains the invariant that, at the conclusion of a scaling phase, the excesses of {\it all} vertices fall below $\Delta/2$. This discharging scheme
takes $O(n^2)$ per scaling phase, yielding an algorithm which takes $O(mn + n^2 \log U)$ time.
Under the assumption that $U = \text{poly}(n)$, their algorithm achieves an $O(mn)$ running time on networks that are non-sparse and non-dense (i.e. where $m = \Theta(n^{1+\epsilon})$).
Our main improvement to the Ahuja-Orlin algorithm is reducing the number of nonsaturating pushes to $O(mn)$ across
all scaling phases, which we succeed in proving for bounded-degree networks. This is detailed in Section \ref{sec:time}.
We conclude this section by stating two results of Goldberg and Tarjan \cite{gold88}. The first is used to prove correctness, while the second is a tool in deriving the time bound. The following lemma is reworded in terms of our terminology.

\begin{lem}[Lemma 3.3 of \cite{gold88}]\label{lem:valid_label} If $f$ is a preflow and $d$ is a distance label obeying the validity condition, then there is no $s-t$ path in the residual network $G_f$. \end{lem}

It is important to note later on that Lemma \ref{lem:valid_label} will not be affected by our change in the admissibility criterion to $d(u) > d(v)$. Since the label $d_h$, which is defined and elaborated on in Section \ref{sec:label},
is relabeled according to the same procedure in the generic push-relabel algorithm, the proof is not altered.

\begin{lem}[Lemma 3.9 of \cite{gold88}]\label{lem:sat_push} The number of saturating pushes is at most $\ell m$, where $d(u) < \ell$, for all $u \in V$. \end{lem}
In the original push-relabel algorithm of Goldberg and Tarjan \cite{gold88}, $\ell \le 2n$. In Section \ref{sec:label}, we will prove that $\ell \le 6n$, and that consequently, our algorithm uses at most $6mn$ saturating pushes.
Again, the proof of Lemma \ref{lem:valid_label} carries over from our result in Section \ref{sec:valid_push_comp}, which shows that the potential function of the generic push-relabel algorithm behaves similarly
with the new admissibility criterion.
%
%

\section{Abundant Arcs}\label{sec:abundant}

In order to reduce the running time of each scaling phase, we classify the set of arcs $A$ into several categories based on their residual capacity. This allows us to disregard arcs that are {\it too small} or
{\it too big} to be of use to us when pushing flow from active vertices. 

We let the {\it compaction capacity} of an arc $(u,v) \in A_f$ be defined as $\gamma(u,v) = r(u,v) + r(v,u)$, where $r(u,v)$ is the residual capacity of $(u,v)$. That is, $\gamma(u,v)$ is the residual capacity {\it between} two vertices.
Recall that $\Delta$ denotes the {\it excess dominator}, which provides an upper bound on the excess during each scaling phase, and that $\Delta_i$ denotes the excess dominator for the $i$th scaling phase. 
We now use the compaction capacity
to describe several categories of arcs. An arc $(u,v)$ is {\it $\Delta$-favorable} if $\Delta/4 < \gamma(u,v) \leq \Delta/2$. It is {\it $\Delta$-large} if $\Delta/4 \leq \gamma(u,v) \leq 2\Delta$, and $r(u,v) \leq \Delta$.
It is {\it $\Delta$-abundant} if $r(u,v) > \Delta$.

The next lemma follows from our improvement property, namely $\Delta_{i+1} \leq \Delta_i/2$ 
\begin{lem}\label{lem:stay_abund} If an arc $(u,v)$ is $\Delta_{i}$-abundant during the $i$th scaling phase, then it will be $\Delta_{i+1}$-abundant, and consequently
$\Delta'$-abundant for every subsequent $\Delta'$-scaling phase. \end{lem}

\section{Compact Networks}\label{sec:compact}
The compact network, built at the beginning of each phase, is a modified version of the residual network, with two specific differences. First, we eliminate non-active vertices incident to arcs of high capacity, and construct high-capacity {\it pseudoarcs}, which represent paths of these high-capacity arcs, in order to perform a sequence of pushes that would have gone through these vertices.
Second, we create low-capacity pseudoarcs, so that we can perform sequences of saturating pushes in one operation. We begin by introducing notation, and then give an algorithm for creating the compact network $G_C$. 

A pseudoarc is {\it $\Delta$-abundant} if it consists entirely of abundant arcs. A pseudoarc is {\it $\Delta$-small} when it contains a non-abundant arc. $V_A$ denotes the set of {\it active vertices}; that is, the vertices $u \in V$ that
at the beginning of the phase have $\Delta/2 < e(u) \leq \Delta$. $V_{SC}$ denotes the set of vertices not in $V_A$ that are incident to $\Delta$-favorable or $\Delta$-large arcs. We let $V_C = V_{SC} \cup V_A$.

The set $A_1$ denotes the set of {\it original arcs} from $A_f$ that are included in the compact network (that is, the $\Delta$-favorable and $\Delta$-large arcs). $A_2$ denotes the set of pseudoarcs. $A_C = A_1 \cup A_2$.

\subsection{Creating the Compact Network}\label{sec:create}
We now describe a procedure for creating the compact network $G_C$ (Algorithm \ref{alg:compact}). Before we can do this, however, we define the {\it abundance graph}; $G^{ab} = (V \backslash V_A,A^{ab})$, where $A^{ab}$ contains all arcs
$(u,v)$ such that $r(u,v) > \Delta$, and the vertex set excludes any $\Delta$-active vertices. We use this subgraph of $G$ to efficiently construct pseudoarcs. The motivation behind using this subgraph 
stems from the fact that we cannot include any active vertex $u \in V_A$ 
along a pseudoarc; consider the situation where {\it create-all-pseudoarcs} iteratively transfers capacity away from a vertex $u \in V_A$, such that there is no original arc left incident to $u \in V_A$ after {\it create-all-pseudoarcs} has terminated.
If this occurs, we will not be able to discharge $u$ at all. In Appendix \ref{app:del_act}, we describe an efficient method  to delete active vertices from the network.

\theoremstyle{definition}
\newtheorem{alg}{Algorithm}
\begin{figure}[H]
\begin{framed}
\begin{alg}\label{alg:compact}
An algorithm for creating compact networks
\end{alg}
\begin{compactenum}[]
	\item {\it Input:} The residual network $G_f$, $\Delta$.
	\item {\it Output:} The compact network $G_C$.
	\item {\bf Step 1a.} Let $A_1$ denote the set containing all $\Delta$-favorable and $\Delta$-large arcs.
	\item {\bf Step 1b.} Let $V_A$ denote the vertices that are active at the beginning of the phase. Let $V_C$ denote vertices
	incident to arcs in $A_1$.
	\item {\bf Step 2a.} Construct $\Delta$-abundant pseudoarcs by calling {\it create-all-pseudoarcs} on $G^{ab}$ with $\rho = \Delta$.
	\item {\bf Step 2b.} Construct $\Delta$-small pseudoarcs by calling {\it create-all-pseudoarcs} on $G_f = (V\backslash V_A,A_f)$ with $\rho = 0$.
	\item {\bf Step 3.} Collect all pseudoarcs in $A_2$ and return $G_C = (V_{SC} \bigcup V_A,A_1 \bigcup A_2)$.
	\item {\bf Step 4.} Generate a valid distance label $d_h$ by calling {\it global-relabel} on $G_C$.
	\item {\bf Step 5.} For each $u \in V_{SC} \bigcup V_A$ initialize $edge\text{-}list[u]$ to contain admissible arcs $(u,v)$, ordered by $d_h(v)$.
\end{compactenum}
\end{framed}
\end{figure}

The algorithms for constructing pseudoarcs are detailed in Appendix \ref{app:tc}.
We first show that the structure of arcs incident to active vertices is preserved between $G_f$ and $G_C$;

\begin{lem}\label{lem:deg_same}
For all $u \in V_A$, $\deg_f[u] = \deg_C[u]$, where $\deg_f$ denotes the degree in the residual network and $\deg_C$ denotes the degree
in the compact network.
\end{lem}
\begin{proof}
This follows trivially from Algorithm \ref{alg:compact}. Only steps 2a and 2b construct pseudoarcs; in both
cases, the set of active vertices $V_A$ is excluded from the algorithm. Therefore, $\deg_f[u] = \deg_C[u]$.
\end{proof}

 \begin{lem}\label{lem:arcs_in} During a $\Delta$-scaling phase, if there is an arc $(u,v)$ where $\Delta/4 < \gamma(u,v) \leq 2\Delta$, at least $(u,v)$ or $(v,u)$
is in the compact network $G_C$ as an original arc. \end{lem}
\begin{proof} Consider the case where $\Delta/4 < \gamma(u,v) \leq \Delta/2$. If this is true, then $(u,v)$ is $\Delta$-favorable, and by Step 1a of the algorithm
it is included in the compact network. If $\Delta/2 < \gamma(u,v) \leq 2\Delta$, then $(u,v)$ will only be in $G_C$ if $r(u,v) \leq \Delta$. But, if neither $r(u,v) \leq \Delta$ nor $r(v,u) \leq \Delta$ then $r(u,v) + r(v,u) > 2\Delta$, which is a contradiction. \end{proof}

Our next lemma shows that the procedure {\it create-all-pseudoarcs} will correctly construct the abundant and small pseudoarcs. We show how to implement these algorithms using dynamic trees in
Appendix \ref{app:tc}. 
\begin{lem}[Capacity Transfer Lemma]
\label{lem:max_bfs} A  push on any pseudoarc ($\Delta$-small or $\Delta$-abundant)  created by the procedure {\it create-all-pseudoarcs} always corresponds to a  push in $G_f$ that does not violate the capacity constraints. \end{lem}
\begin{proof} 
The only possible difficulty would be
the case where pseudoarcs ``share'' an arc. More formally, given $k$ pseudoarcs $(u_1,w_1) = \ang{u_1,v_1,v_2,\ldots,v_k,w_1}$, $(u_2,w_2) = \ang{u_2,v_1,v_2,\ldots,v_k,w_2}$, 
$\ldots $ \\$(u_k,w_k)$,
suppose that the arc $(v_i,v_{i+1}) \in (u_1,w)$ is also contained in $(u_2,w_2) \ldots (u_k,w_k)$. 

In the procedure {\it transfer-capacity} in Appendix \ref{app:tc}, once we construct a pseudoarc,
we transfer capacity to the pseudoarc and then reduce the path capacity for subsequent pseudoarcs to be constructed. In fact, the multiple pseudoarcs created by this procedure
that share one or more internal arcs will still correspond to pushes in $G_f$ that do not violate capacity constraints. The arc $(v_i,v_{i+1})$ may have a reduced capacity on $(u_2,w_2)$ after $(u_1,w_1)$ was constructed; this
corresponds to the situation where a push along the sequence of paths in $(u_1,w_1)$ occurs {\it before} any push from $u_2, \ldots, u_k$. This will correspond to a possible sequence of $k$ pushes in the original residual network,
according to the decreasing residual capacity of the shared arc $(v_i,v_{i+1})$.
\end{proof}

Finally, we will bound the number of $\Delta$-favorable and large arcs in the compact network.
\begin{lem}
\label{lem:fav_in_comp} A $\Delta_i$-favorable arc $(u,v)$ during the $i$th scaling phase appears in the compact network at most 3 times. Moreover, it will be $\Delta_{i+3}$-abundant in 
4 scaling phases.\end{lem}
\begin{proof}
If $(u,v)$ is $\Delta_i$-favorable, then we know $\Delta_i/4 < \gamma(u,v) \leq \Delta_i/2$. By our improvement property, $\Delta$ decreases by a factor of at least 2 during each scaling
phase. Applying the property, we get $\Delta_{i+1} \leq \Delta_i/2$. In the $(i+1)$st scaling phase, we can bound the compaction capacity of $(u,v)$ as $\Delta_{i+1}/2 < \gamma(u,v) \leq \Delta_{i+1}$.
Let $\Delta_{i+3} \leq \Delta_{i+1}/4$. Thus, we can bound $\gamma(u,v)$ as $2\Delta_{i+3} < \gamma(u,v) \leq 4\Delta_{i+3}$, which is clearly $\Delta_{i+3}$-abundant. Therefore, any $\Delta_i$-favorable
arc $(u,v)$ remains in the compact network at most 3 phases, after which it becomes $\Delta_{i+3}$-abundant, and by Lemma \ref{lem:stay_abund}, remains $\Delta'$-abundant for every subsequent excess dominator $\Delta'$.
\end{proof}
Recall that, trivially, once an arc is $\Delta_i$-abundant, it is not $\Delta_i$-favorable, and will not be $\Delta_{i+1}$-large, so it will appear in $G_C$ in at most one more phase. This bound becomes important when bounding the size of the compact network across all phases. But before we are able
to accomplish this, we must describe a scheme for discharging active vertices (i.e. those with excess $\Delta/2 < e(u) \leq
\Delta$).

\subsection{Discharging Active Vertices}\label{sec:disc}

Now that we have bounded the number of $\Delta$-favorable and large arcs in the compact network, we must bound the number of active vertices included at the start of a scaling phase (that is, vertices
that begin the $\Delta$-scaling phase with $\Delta/2 < e(u) \leq \Delta$. We first define the basic {\it initialize} and {\it push} procedures, and then the {\it discharge} procedure. 

Recall that $d_h$ and $d_\ell$ are two distinct vectors that both contribute to the distance labeling of each $u \in V_A$, and that for all $u \in V_A$, $d(u) = d_h(u) + d_\ell(u)$. In the {\it initialize} procedure,
we will use $d_h$ (since we maintain the invariant that $d_h$ is valid within a $\Delta$-scaling phase).

\begin{pseudocode}{initialize}{$G$}
\item Initialize a new preflow $f$ and distance labels $d_h$, $d_\ell$;
\item $d_h(s) := n$ \textbf{and} $d_h(u) := 0$, \textbf{for all} $u \in V \backslash \set{s}$;
\item \textbf{for all} $v \in V \backslash \set{s}$ \textbf{do}
\item \tab \textbf{if} $(u,v) \in A_f$ \textbf{then} $f(u,v) := c(u,v)$;
\end{pseudocode}

\begin{pseudocode}{push}{$u,v$}
\item \comment{\textbf{Applies when} $e(u) > 0, d(u) > d(v)$}
\item $\delta = \min\set{e(u),r(u,v),\Delta-e(v)}$;
\item $f(u,v) = f(u,v) + \delta$;
\item $f(v,u) = f(v,u) - \delta$;
\end{pseudocode}

We now can describe the {\it discharge} procedure. It will (1) {\it push} or {\it relabel} $u$ until $e(u) < \Delta/2$ for $u \in V_{SC}$, and (2) until $e(u) = 0$ for $u \in V_A$. This latter condition ensures that an active vertex in phase $i$ can only appear in $G_C$ in phase $i+1$ as a member of $V_{SC}$, i.e., because it is adjacent to a $\Delta$-favorable edge. We present details in  Lemma \ref{lem:act_in_comp}.

We maintain the invariant throughout the algorithm that within a $\Delta$-scaling phase, there is only 1 {\it low-capacity} nonsaturating push (i.e. that sends $\delta \leq \Delta/2$ units of flow) for each value of the
distance label $d_\ell$, for $u \in V_A$. We enforce
this by {\it incrementing} the distance label $d_\ell$, and then making the push. We ensure that there is only one low-capacity nonsaturating push for each value of $d_\ell$ by utilizing two data structures; a boolean list $nonsat_u[\cdot]$ that keeps track of nonsaturating pushes for each value of $d_\ell$.
Second, we maintain a current edge list $edge\text{-}list[u]$ that orders the admisisble arcs adjacent to $u$ by $d_h$. This way, we do not alter the order in which pushes would have been made in the generic push-relabel algorithm, across any sequence of increments to $d_\ell$.
These are both maintained as global variables for the duration of a $\Delta$-scaling phase.
We let $\delta$ denote the amount of flow sent by a {\it push} operation. 

\begin{figure}[H]
\begin{framed}
\begin{alg}The algorithm for discharging vertices during a $\Delta$-scaling phase.\end{alg}
\begin{pseudocode}{discharge}{$G_C,u$}
\item \textbf{if} $u \in V_A$ \textbf{then}
\item \tab \textbf{while} $e(u) > 0$ \textbf{do}
\item \tab\tab \textbf{if} $e(u) > \Delta/2$ \textbf{then }{\it push} or {\it relabel} $u$;
\item \tab\tab \textbf{else if} $e(u) \leq \Delta/2$ \textbf{then}
\item \tab\tab\tab \textbf{if} $d_\ell(u) = 4n-1$ \textbf{then}
\item \tab\tab\tab\tab {\it push} or {\it relabel} $u$ until $e(u) = 0$;
\item \tab\tab\tab {\bf else if} $edge\text{-}list[u] = \emptyset$ \textbf{then} {\it relabel} $u$;
\item \tab\tab\tab \textbf{else if} $nonsat_u[d_\ell(u)] = true$ \textbf{and}
\item \tab\tab\tab\stab $r(edge\text{-}list[u]) > \delta$ \textbf{then}
\item \tab\tab\tab\tab $d_\ell(u) := d_\ell(u) + 1$;
\item \tab\tab\tab\tab $nonsat_u[d_\ell(u)] := false$;
\item \tab\tab\tab\tab Add new admissible edges to
\item \tab\tab\tab\tab\stab to $edge\text{-}list[u]$ ordered by $d_h$;
\item \tab\tab\tab \textbf{while} $nonsat_u[d_\ell(u)] = false$ \textbf{do} {\it push} on 
\item \tab\tab\tab\stab $(u,v) := dequeue(edge\text{-}list[u])$ or {\it relabel} $u$;
\item \tab\tab\tab \textbf{if} $edge\text{-}list[u] = \emptyset$ \textbf{then} {\it relabel} $u$;
\item \tab\tab\tab\stab \textbf{until} there is a nonsaturating push;
\item \tab\tab\tab $nonsat_u[d_\ell(u)] := true$;
\item \textbf{else if} $u \not\in V_A$ \textbf{then}
\item \tab \textbf{while} $e(u) > \Delta/2$ \textbf{do}
\item \tab\tab {\it push} or {\it relabel} $u$;
\end{pseudocode}
\end{framed}
\end{figure}

The next lemma follows directly from the definition of a $\Delta$-favorable arc and the {\it discharge} procedure;
\begin{lem} A push along any $\Delta$-favorable arc $(u,v)$, when $u \in V_{SC}$ will always be saturating. \end{lem}
\begin{proof}
If $u \in V_{SC}$ during the $\Delta_i$-scaling phase, then no push originating from $u$ will send less than $\Delta/2$ units of flow; this follows directly from {\it discharge}, since we do not require any $u \not\in V_A$ to satisfy $e(u) = 0$ at the
end of the $\Delta_i$-scaling phase. Thus, we only require that $e(u) < \Delta_k/2$ at the end of the phase. Now, recall that an arc $(v,w)$ is $\Delta_i$-favorable, then $\Delta/4 \leq \gamma(u,v) \leq \Delta/2$. Thus, the capacity of $(v,w)$
during the execution of a scaling phase is at most $\Delta/2$. Since each push from $u$ sends at least $\Delta_i/2$ units of flow, every $\Delta_i$-favorable arc will be saturated by a push from $u \in V_{SC}$.
\end{proof}

In every $\Delta$-scaling phase, we iteratively call {\it discharge} until we can proceed to the next scaling phase. Note that for some $u \in V_A$, when the excess is $e(u) \leq \Delta/2$, we modify the manner in which we
discharge vertices. For reasons that become clear in the analysis of the algorithm, we divide nonsaturating pushes into {\it high-capacity} and {\it low-capacity}, when the flow $\delta$ sent is $\delta > \Delta/2$ or $\delta \leq \Delta/2$,
respectively.

Note that we limit $d_\ell(u)$ to have the maximum value $4n-1$. For each $u \in V_A$, we only permit one low-capacity nonsaturating push per value of $d_\ell$, until $d_\ell(u) = 4n-1$. Then, we allow a second group of nonsaturating pushes to be made. We
will show that there are at most $O(n)$ nonsaturating pushes per group of pushes, and that this only happens twice during the execution of a scaling phase (Lemma \ref{lem:act_correct}). We state this formally in the lemma below;

\begin{lem} There is at most one low-capacity nonsaturating push (that is, a push that sends $\delta \leq \Delta/2$) for each value of $d_\ell(u)$, for each $u \in V_A$, until $d_\ell = 4n-1$, after which
there will be only $2n$ more such pushes. \end{lem}

The next lemma follows the approach of Lemma 5 from \cite{ahuja89};

\begin{lem}\label{lem:disc1_cond} The algorithm obeys the following conditions: 
\begin{compactenum}[\upshape \hspace{4mm}\bfseries C1:]
\item Each nonsaturating push from an active vertex from $u \in V_{SC}$ will send at least $\delta > \Delta/2$ units of flow.
\item No excesses rise above $\Delta$. 
\end{compactenum}
\end{lem}
\begin{proof}
Consider a push on $(u,v)$. If $u \in V_{SC}$, then $e(u) > \Delta/2$., and $e(v) \leq \Delta/2$, since $u$ is the vertex
with the minimum distance label among vertices in $V_{SC}$ such that $e(u)$ is active. Since $d(u) > d(v)$ as well by our
admissibility criterion, we send $\delta = \min\set{e(u),r(u,v),\Delta-e(v)} \geq \min\set{\Delta/2,r(u,v)}$. In the case that
$\min\set{\Delta/2,r(u,v)} = \Delta/2$, the push will be nonsaturating, but will send $\delta \geq \Delta/2$. This concludes
the proof of {\bf C1}. {\bf C2} follows immediately from line 02 in the {\it push} procedure, which will $\delta \leq \Delta-e(v)$
for any push along $(u,v)$.
\end{proof}

Next, we will prove two lemmas that ensure each $u \in V_A$ can be discharged at the end of a $\Delta$-scaling phase. We start with a technical lemma regarding path decompositions
in the residual network, and then bound $|V_A|$ across all scaling phases.

\begin{lem}[Path Decomposition Lemma]\label{lem:decomp}
For each $u \in V_C$ such that $e(u) > 0$, the $s-u$ preflow can be decomposed into a collection of paths $\mathcal{P}$
such that $e(u)$ units of flow can be returned to $s$ along paths $p \in \mathcal{P}$. Moreover,
	$$\mathcal{C} = \sum_{p \in P} \min_{(i,j) \in p} r(i,j) \geq e(u).$$
That is, the sum of the minimum-capacity arc in each path $p$ in the decomposition $P$ yields a capacity of at least $e(u)$.
\end{lem}
\begin{proof}
The proof relies on the construction of a new flow network, which we denote as $G_{s,u}$. The construction is as follows;
consider all paths $s \rightarrow u$ in $G_C$. Denote the set of these paths as $P_{s,u}$. Now, let $V(G_{s,u})$ consist
of the set of all vertices contained in paths $p \in P_{s,u}$, and similarly let $A(G_{s,u})$ consist of the arcs from these paths.

Now, let the source vertex of $G_{s,u}$ be $s' = u$, and let the sink $t' = s$. Finally, for each $v \in V(G_{s,u}) \backslash
\set{s',t'}$, let $e_{s,u} = 0$ (where $e_{s,u}:V \rightarrow \mathbb{Z}_{\geq 0}$ denotes the excess for $G_{s,u}$). Thus,
we eliminate the excess on all vertices in $G_{s,u}$. We remark that there is an $s'-t'$ path in $G_{s,u}$ since
$e(u)>0$ in $G_C$; when $e(u)$ is positive, there is always a path from $u$ to $s$. This ensures the existence of an $s'-t'$ path.

Finally, we can apply the Flow Decomposition Theorem of Ford and Fulkerson to $G_{s,u}$. That is, we know that
there exists a set of feasible flows $f_1,\ldots,f_k$ associated with $s'-t'$ paths $p_1,\ldots,p_k$ such that $\forall i$,
$f_i$ sends positive flow only on $p_i$. This guarantees the existence of paths $\mathcal{P} = \set{p_i \mid 1 \leq i \leq k}$
such that the excess from $u$ in the compact network $G_C$ can be returned to $s$ along paths in $\mathcal{P}$.

Moreover, we show that the second claim holds; we will prove that $e(u)$ places a lower bound on the sum of the capacities 
of the bottleneck (minimum-capacity) arcs for each $p \in P$. An arc $(u,v) \in A_f$ is a {\it bottleneck arc} on a path $p$ if $r(u,v) = \min_{(i,j) \in p} r(i,j)$. When this arc is
saturated, $p$ is sending the maximum amount of flow possible from $s$ to $u$.
When $\mathcal{C} = e(u)$, we see that all bottleneck arcs are saturated, since the preflow $f$ obeys the
capacity constraint. When $\mathcal{C} > e(u)$, then some bottleneck arcs are left unsaturated. This guarantees that there always exists a path $p$ or a collection of paths
$p_1 \ldots p_k$ such that flow can be returned to the source from $u$ with sufficient relabelings (sufficient relabelings will make paths from the decomposition admissible, so
therefore flow can be returned to $s$).
\end{proof}

Since Lemma \ref{lem:decomp} holds, we can immediately see that flow can be returned to the source from each active vertex after sufficient relabelings. Next, we two crucial
lemmas that guarantee the algorithm's correctness. We first prove that we can efficiently order the current edge list by $d_h$, and then that we are able to discharge all
active vertices at the conclusion of a $\Delta$-scaling phase.
\begin{lem}\label{lem:edge_order}
For each $u \in V_C$ during a $\Delta$-scaling phase, $edge\text{-}list[u]$ is in increasing order in terms of $d_h$. Formally, $d_h(u_1) \leq d_h(u_2) \leq \ldots \leq d_h(u_k)$ for
each of the $k$ incident arcs on $u$.
\end{lem}
\begin{proof}
We can use the buckets data structure described in Ahuja and Orlin \cite{ahuja89} to initially order the current edge list when we construct $G_C$. What is left to show is that
$edge\text{-}list[u]$ remains ordered after a sequence of edge additions. We can maintain this as a priority queue keyed by $d_h(v)$, for each $v \in edge\text{-}list[u]$. 
\end{proof}

We remark that additions can be accomplished in $O(1)$ time, due to the fact that $k$, the maximum in-degree, is constant. Thus, we can implement this with a simple buckets data structure
described by Ahuja and Orlin \cite{ahuja89}. This will cost $O(mn)$ across all phases. Next, we prove that discharge fulfills its promise, namely that it correctly discharges each active vertex $u \in V_A$, so that $e(u) = 0$ at the termination of a phase.

\begin{lem}\label{lem:act_correct}
For each $u \in V_A$, $e(u) = 0$ at the conclusion of the $\Delta$-scaling phase.\end{lem}
\begin{proof}
We first make a remark about an invariant that holds for the duration of a ``discharge'' on some $u \in V_C$; once an arc 
$(u,v)$ incident to $u$ becomes admissible, it will {\it stay} admissible while $u$ is selected by discharge (since $d(u) > d(v)$,
not $d(u) = d(v) + 1$). Therefore, despite the increments to $d_\ell$, we will still be able to select arcs from $edge\text{-}list[u]$ in  order to
send flow on, even if they become admissible under a previous increment.

By Lemma \ref{lem:edge_order}, we know that the current edge list is in increasing order by $d_h$.
 Since $d_h$ satisfies the validity condition (Lemma \ref{lem:h1}), it
provides an estimate 
of the distance from any $u$ to the sink $t$. The ordering of $edge\text{-}list$ becomes important in the following case; consider a vertex $u$ that is selected by {\it discharge}. Then there may
be an edge $(u,s)$ within $G_C$ that can become admissible as soon as $d(u) = 2n+1$. However, due to the ordering of the edges, a push along $(u,s)$ will come {\it after} all edges that
have a chance of sending excess ``forward'' (i.e. to the sink $t$) before it is returned to $s$. This is important because of our altered labeling $d = d_h + d_\ell$, we may relabel
at different times. Thus, we wish to avoid returning flow to the source $s$ prematurely.

The final component of the correctness comes from the two groups in
which we complete small nonsaturating pushes. We first make one low-capacity nonsaturating push for each value of $d_\ell$, while $d_\ell(u) \leq 4n-1$. Afterwards, we will discharge $u$ so that $e(u) = 0$
after it reaches the maximum label, ignoring the limitation of one low-capacity nonsaturating push per value of $d_\ell$. However, since $\deg[u] = O(1)$ and there are only $2n$ distinct values of $d_h(u)$ from which $d(v)$ could
exceed $d(u)$, there will be at most $2n$ such small nonsaturating pushes made. Therefore, we conclude that all $u \in V_A$ are completely discharged.
\end{proof}

\theoremstyle{definition}

The previous two lemmas underscore the importance that the maximum degree of a vertex is a constant $k$. This permits us to efficiently maintain the ordering of the current-edge list,
so we always send flow towards the sink $t$ before the source $s$. If we allowed general networks, where $\max\set{\deg[u]|u \in V} = O(n)$, then each push would potentially
take $O(\log n)$ time due to the nature of the priority queue implementation and worst-case bounds on the data structure.

We now state a main result of our paper regarding the number of active vertices in the compact network.

\begin{lem}[Active Vertices Lemma]
\label{lem:act_in_comp} For a vertex $u \in V_C^i$ such that $u$ is $\Delta_i$-active, if  $u \in V_C^{i+1}$  then $u$ is incident to a $\Delta_{i+1}$-large or $\Delta_{i+1}$-favorable arc.\footnote{We denote the set
of vertices in the compact network during the $i$th scaling phase as $V_C^i$.}\end{lem}
\begin{proof}
If some vertex $u \in V_C^i$ is active during the $\Delta_i$-scaling phase, then by the {\it discharge} procedure, $e(u) = 0$ will be true when the phase terminates. That is, $u$ will not be active in the $\Delta_{i+1}$-scaling
phase. The only way $u \in V_C^{i+1}$ holds is if $u \in V_{SC}^{i+1}$; that is, it was included as an endpoint of a $\Delta_{i+1}$-favorable or large arc. Thus, we can see that it will only take one phase for an active
vertex $u \in V_A$ to become a vertex contained in the set $V_{SC}$ in the next scaling phase. 
\end{proof}

Note that if we did not {\it completely} discharge vertices that were active at the start of a $\Delta$-scaling phase, then we could have $|V_C| = \Omega(n)$ active vertices in the compact network during {\it each}
scaling phase. Rather, Lemma \ref{lem:act_in_comp} allows us to bound the number of vertices in the compact network across all scaling phases in Theorem \ref{thm:om_comp}. Further, since
we shift flow along each vertex within a pseudoarc, proof of the following lemma is immediate.

\begin{lem}
A push along a pseudoarc $(u,w) = \ang{u,v_1,v_2,$
$\ldots,v_k,w}$ will not alter the excesses of any internal vertex $v_1 \ldots v_k$ that is included within the pseudoarc.
\end{lem}

This guarantees that no vertex that appears on a pseudoarc will become active as a result of a pseudoarc push.

\subsection{Size of the Compact Network}\label{sec:size}
\theoremstyle{plain}
\newtheorem{thm}{Theorem}[section]
\begin{thm}[Size of the Compact Network]\label{thm:om_comp} There are $O(m)$ vertices in the compact network across all scaling phases.\end{thm}
\begin{proof}
This theorem immediately follows from Lemmas \ref{lem:fav_in_comp} and \ref{lem:act_in_comp}. In particular, Lemma \ref{lem:act_in_comp} shows that we can convert an active vertex in the $\Delta_i$ scaling phase
to a vertex in $V_{SC}$ in the $\Delta_{i+1}$ scaling phase; it remains active for $O(1)$ phases. Thus, the theorem holds.
\end{proof}

This theorem underscores the significance of our algorithm; superficially, it may seem that we maintain paths of ``medium'' and ``large'' arcs so that we can conduct sequences of saturating and/or nonsaturating pushes. However,
the idea of compaction is what leads to our improved running time. Because we can bound the total number of vertices that will be eligible for flow operations 
at $O(m)$ across all phases, we are able to reduce the number of nonsaturating pushes to $O(mn)$ across
all scaling phases. An important distinction between our algorithm and Orlin's \cite{orlin13} is that we do not explicitly ``compact'' active vertices $u \in V_A$ out of the network after $O(1)$ scaling phases. Instead, we show that
they are ``converted'' (in a sense) to a vertex $u \in V_{SC}$ if they remain in the network after they are discharged. This correspondence, while not permitting us to bound the number of active vertices at $O(n)$, does match
the number of vertices in $V_{SC}$, and still shows that the network is indeed compact.

\subsection{Labelings in Compact Networks}\label{sec:label}
In this section, we describe how to maintain a valid labeling during each $\Delta$-scaling phase. We resolve two difficulties, namely maintaining a valid labeling between $G_C$
and $G_f$. To accomplish this, we use a technique due to Goldberg, et. al. \cite{gold97,gold97-2}, which generates a valid labeling from a BFS-tree. Next, we show how we use distinct
labelings for high-capacity nonsaturating pushes ($d_h$), and for low-capacity nonsaturating pushes ($d_\ell$) to count the number of nonsaturating pushes.

\begin{pseudocode}{relabel}{$G,u$}
\item \comment{\textbf{Applies when} $\not\exists$ admissible arc $(u,v) \in A_C$}
\item $d_h(u) := \min\set{d_h(v) | (u,v) \in A} + 1$;
\item {\bf if} $d(u) < d(v)$ {\bf then} ${d_\ell(u) := d_\ell(v)}$;
\end{pseudocode}

The {\it relabel} procedure in our algorithm contains a few notable departures from the analagous procedure in \cite{gold88,ahuja89}. We note that in lines 02, $d(u)$ may still be less than
$d(v)$ after a relabeling. This is due to the fact that $d_\ell(v)$ could have increased a large amount when compared to $d_\ell(u)$, and even though $d_h(u) = d_h(v) + 1$, the discrepancy exists
in the overall labeling $d$. We resolve this by performing a ``gap relabel" on $d_\ell$ if $d(u) < d(v)$ (line 03 of the {\it relabel} procedure).

The next lemma shows that $d_h$ is bounded above by $2n$, for all $u \in V$, and is proven in \cite{gold88}.

\begin{lem}\label{lem:h1} For any $u \in V$, $d_h(u) < 2n.$\end{lem}
\begin{proof}
Since we know that there exists a path from $u$ to $s$ for each $u \in V_A$, we consider the case that gives us the maximum label. Recall that we only modify by additions of 1 ($d_h(u) = d_h(v) + 1$). 
Thus, $d_h$ remains a valid labeling during the execution of a $\Delta$-scaling phase. A simple path has length at most $n-1$, so therefore we have $d(s) = n$, so $\max d_h = 2n-1$. 
\end{proof}

Since we generate a valid distance labeling at the beginning of the $\Delta$-scaling phase with {\it global-relabel}, and only increase $d_h$ when there are no incident admissible arcs, this will behave in the same manner
as the labeling in the generic push-relabel algorithm \cite{gold88}, as well as the excess-scaling algorithm \cite{ahuja89}. The distance label will only increase by 1 due to each {\it relabel} operation, so proof of the lemma
still follows.

Next, we must consider increments in the distance label from {\it discharge} that result from a low-capacity nonsaturating push.
The following lemma will show that $d_\ell$ also will not increase above $2n$.

\begin{lem}\label{lem:h2} For all $u \in V$, there is an increase in the overall distance $d(u)$ by $4n-1$ due to both gap-relabelings and low-capacity nonsaturating pushes. \end{lem}
\begin{proof}
Each increment in $d_\ell$ increases the label by 1. When we need to make an additional nonsaturating push, the {\it discharge} procedure will increment the label. Since
we limit $d_\ell(u)$ at $4n-1$, it is clear that this value bounds the increase.
\end{proof}

The previous two lemmas allow us to state a final conclusion that bounds the total increase in the overall label $d$.

\begin{lem}\label{lem:height_bound} For each $u \in V$, $d_h < 2n$, and the increase to $d_\ell$ due to low-capacity nonsaturating pushes is $4n$.
Moreover, there are $6n^2 = O(n^2)$ relabelings throughout the algorithm. \end{lem}
\begin{proof} Proof is immediate from Lemmas \ref{lem:h1} and \ref{lem:h2}. \end{proof}

It is important to remark that even though we maintain two distinct vectors ($d_h$ and $d_\ell$) for the labeling, the algorithm's correctness remains. More formally, we can ensure that we will
never make a nonsaturating push such that the value of the potential function is increased. This is due to our initial definition of $d(u) = d_h(u) + d_\ell(u)$. We will allow a push to occur only if
$d(u) > d(u) \Rightarrow d_h(u) + d_\ell(u) > d_h(v) + d_\ell(v)$. We will {\it relabel} $u$ until an arc incident to $u$ is admissible, an operation which is bounded at $2n$. Therefore, the potential
function will still decrease with each nonsaturating push, and correctness remains.

Next, we describe a variant of the global relabel procedure of Goldberg. This algorithm relies on breadth-first search to traverse the graph and 
generate a labeling based on the levels in the BFS tree.

\begin{pseudocode}{global-relabel}{$G$}
\item Run reverse BFS from $t$ and let $d_t$
\item \stab denote the distances from the BFS tree;
\item Run BFS from $s$ and let $d_s$ denote the
\item \stab distances from the BFS tree;
\item {\bf for all} $u \in V$ {\bf do} 
\item \tab $d_h(u) := \min\set{d_s(u)+n,d_t(u)}$;
\item \tab $d_\ell(u) := 0$;
\item {\bf return} $d$;
\end{pseudocode}

We will use the {\it global-relabel} procedure twice. Once we build the network, we generate a valid labeling, incorporating psuedoarcs. Then, we can simply push and relabel vertices as in the generic algorithm. The second time we apply it is when transforming the compact preflow $f$ into the residual preflow $f'$. The {\it global-relabel} procedure costs $O(m)$ time per scaling phase, since we run BFS twice. Now, we prove that the algorithm maintains a valid labeling across all scaling phases. 

\begin{thm}[Valid Labeling]\label{thm:valid_label} Within each $\Delta$-scaling phase, $d_h$ obeys the validity condition. Moreover, at the end of each phase, we generate a labeling for $G_f$ that obeys the validity condition. \end{thm}
\begin{proof}
First, consider initialization. Since $d(s) = n$ and $d(u) = 0$ for each $u \in V \backslash \set{s}$ when the algorithm starts, the labeling is trivially valid. Within each $\Delta$-scaling phase,
we claim that $d_h$ obeys the validity condition (i.e. $\forall (u,v) \in A$, $d(u) \leq d(v) + 1$). We clearly see that within {\it relabel}, we will increase $d_h(u)$ by at most $d(v) + 1$. Thus, each
invocation of {\it relabel} maintains the validity of $d_h$. Applying this argument inductively, we see that $d_h$ remains valid across a sequence of {\it push} and {\it relabel} operations. 

Furthermore, in {\it global-relabel}, we see that if we retrieve labels directly from the BFS tree, then validity is immediate (this is also discussed in \cite{gold97-2}). This ensures that all forward arcs $(u,v) \in A$
are labeled such that $d(u) \leq d(v) + 1$. This means that the lableing $d_h$ in the compact network is valid while we are running a $\Delta$-scaling phase, and further, the labeling in the residual network
is valid once we have transformed the compact preflow into the residual preflow.
\end{proof}

\section{Validity of Pushes in the Compact Network}\label{sec:valid_push_comp}
When flow is sent along a pseudoarc $(u,w) \in A_2$, there may be pushes that are inadmissible (that is, pushes that send flow from $v_i$ to $v_{i+1}$ when $d(v_i) \leq d(v_{i+1})$). In this section, we show that as long as flow is
{\it eventually} sent to a lower-labeled vertex, inadmissible pushes within pseudoarcs do not affect termination.

In order to accomplish this, we will consider a generalized version of the original push-relabel algorithm of Goldberg and Tarjan \cite{gold88}. We use the potential function $\Phi_g$, given below, to show that so long as flow is shifted
to a lower-labeled vertex, the potential function behaves exactly the same as in the original push-relabel algorithm.

Let $\Phi_g$ be a potential function defined as follows;
	\begin{equation*} \Phi_g = \sum_{u \in V:e(u) > 0} d(u) \end{equation*}

The following lemma shows that the behavior is indeed what we claim; namely, that relaxing the admissibility constraint to $d(u) > d(v)$ still results in correct behavior of the potential function.

\begin{lem}\label{lem:gen_pr} If we relax the definition of admissibility to $d(u) > d(v)$, then the generic push-relabel algorithm will terminate. \end{lem}
\begin{proof}
We first examine the case where the potential function will increase. $\Phi_g$ increases when there is a relabeling or a saturating push. By Lemma \ref{lem:height_bound}, we see
that there are $6n^2$ relabelings. Similarly, we can state that there will be $6mn$ saturating pushes. Therefore, $\Phi_g$ will increase by $6n^2m$.

Now we consider the case where $\Phi_g$ decreases. Clearly, the potential decreases by at least 1 when a push along an original arc is made. Now, consider a simple path
$P=\ang{u,v_1,v_2,\ldots,v_k,w}$, where $d(u) > d(w)$. We will consider a push along the entirety of $P$. We see that $\Phi_g$ will decrease by at least 1 as well,
since flow is still shifted to a vertex with a lower distance label. So long as we enforce $d(u) > d(w)$, the termination of the generic push-relabel algorithm still holds
under path pushes. Moreover, we will terminate with $O(n^2m)$ nonsaturating pushes; this shows that the same running-time bounds apply to this stronger form of
the generic push-relabel algorithm as the original in \cite{gold88}, and that the algorithm terminates under path pushes.
\end{proof}

It is important to note that the previous lemma demonstrates that pseudoarc pushes cause the potential function from the original Goldberg-Tarjan push-relabel algorithm behaves the same if path pushes are permitted, whereas
Lemma \ref{lem:act_correct} simply shows that the {\it discharge} procedure is correct.

We have shown a stronger form of the generic push-relabel algorithm of Goldberg and Tarjan. So long as flow is moved to a lower-labeled vertex, $\Phi_g$ behaves
as it does in the generic push-relabel algorithm. 
This gives us the termination of a $\Delta$-scaling phase;

\newtheorem{corr}{Corollary}[section]
\begin{corr} If pushes along directed paths are allowed, each $\Delta$-scaling phase will still terminate.\end{corr}

This proof relies on the assumption that we bound the relabel operations at $O(n^2)$ per phase. A difficulty that we resolve in Section \ref{sec:label} is due to the admissibility criterion; since $d(u) > d(v)$, the bound
on $2n-1$ for vertex labels from \cite{gold88} no longer applies. We show in the next section that by imposing several constraints on relabeling operations, we can bound the distance labels at $4n-1$ for each $u \in V$,
and at $4n^2 = O(n^2)$ across the vertex set.

\section{Analysis of the Algorithm}\label{sec:time}
\subsection{Bound on Saturating \& Nonsaturating Pushes}
\begin{lem}There are $O(mn)$ saturating pushes across all scaling phases.\end{lem}
\begin{proof}Proof is immediate from Lemma \ref{lem:sat_push}. \end{proof}
\begin{lem}\label{lem:nonsat_big}There are $16cn$ large nonsaturating pushes per scaling phase, and $O(mn)$ across all scaling phases.\end{lem}
\begin{proof}
Consider the potential function
\begin{equation*}\label{eq:aopot} 
\Phi = \sum_{u \in V_C} \frac{[d_h(u) + d_\ell(u)]\cdot e(u)}{\Delta}. \end{equation*}

We will first consider $\Phi$ at the beginning of a scaling phase. Since $d_h(u) < 2n$, $d_\ell(u) = 0$ at the start of a scaling phase, and $e(u) \leq \Delta$, for all $u \in V$, we see
that
$\Phi_\text{init} =  2cn.$ 

Next, we see that $\Phi_\text{max} = 8cn$. This is immediate from Lemma \ref{lem:height_bound}, since there is an increase of at most $6n$ in $d_h + d_\ell$.

When a nonsaturating (or saturating) push occurs, $d_h(u) + d_\ell(u) > d_h(v) + d_\ell(u)$ must hold for the push to have been made, the overall
cost of distance labels in $\Phi$ will decrease.
Since we can guarantee that a large nonsaturating push will send at least $\delta \geq \Delta/2$ (Lemma \ref{lem:disc1_cond}),
$\Phi$ will decrease by $1/2$. Therefore, there will be $16cn$ large nonsaturating pushes per scaling phase.

Extending this result across all scaling phases, we have
	\begin{align*} \sum_{i=1}^K 12cn &= \sum_{i=1}^K O(c_in) = O(n) \cdot (c_1 + c_2 + \cdots + c_K)  \\
	\label{eq:om}&= O(n) \cdot O(c_1+c_2+\cdots+c_K) = O(mn). \end{align*}

The final statement is given to us by Theorem \ref{thm:om_comp}.  clearly, we have $O(mn)$ large nonsaturating pushes across all scaling phases,
and the proof holds.
\end{proof}

We must finally bound the number of small nonsaturating pushes. Our maintenance of the list $nonsat_u[\cdot]$ ensures that we only
permit a single small nonsaturating push per value of $d_\ell$ in {\it discharge}. Moreover, once $d_\ell(u) = 4n-1$, for some $u \in V_A$, we
permit it to make as many nonsaturating pushes as are necessary; we will see that this quantity is $2n$ in the next lemma. This lemma
makes use of the fact that the maximum in-degree of a vertex is $k = O(1)$. If general networks were permitted, the overall cost of this
procedure may be as large as $O(mn^2)$ across all scaling phases; however, it is efficient in the bounded-degree case. 

\begin{lem}\label{lem:small_nonsat_om} There are at most $4n-1$ low-capacity nonsaturating pushes while $d_\ell(u) < 4n-1$, and at most $2n$ such pushes once $d_\ell(u) = 4n-1$.\end{lem}
\begin{proof}
We enforce the condition that there is a single low-capacity nonsaturating push for each value of $d_\ell$; therefore, it is immediate that when $d_\ell(u) < 4n-1$, 
there are at most $4n$ low-capacity nonsaturating pushes. Finally, we must consider the quantity of such pushes when $d_\ell(u) = 4n-1$; there
are at most $2n$ distinct values $d_h$ can attain that could cause some $d(u) > d(v)$. Since $\forall u \in V_A$, $\deg[u] = O(1)$, there
will be $2n$ such pushes per vertex. Therefore, there are at most $2n$ low-capacity nonsaturating pushes once
$d_\ell(u) = 4n-1$, and $6n$ low-capacity nonsaturating pushes overall. 

By Lemma \ref{lem:act_in_comp}, we know that there are $O(m)$ active vertices across all scaling phases. Summing across all scaling phases gives us that
there are $6mn = O(mn)$ low-capacity nonsaturating pushes.
\end{proof}

\begin{thm}\label{thm:nonsat_bound_all} There are $O(mn)$ nonsaturating pushes across all $\Delta$-scaling phases. \end{thm}

\subsection{Time to Create Compact Networks \& Transform the Compact Preflow}
In Appendix \ref{app:tc}, we show how we can both construct the compact network and transform the compact preflow into a residual preflow in $O(m \log n)$ time using
the dynamic trees data structure. We restate a theorem proven later that summarizes these bounds.
\begin{thm}[Theorem \ref{thm:dtree}] \label{thm:dtree_prev}
Constructing abundant and small pseudoarcs takes $O(m \log n)$ time per scaling phase by calling {\it create-all-pseudoarcs.} It takes $O(m \log n)$ time to transform the compact preflow to the
residual preflow by calling {\it restore-all-flows}.
\end{thm}

\section{A Strongly-Polynomial Variant}\label{sec:poly}

In this section, we define our final algorithm, {\it max-flow-1}, and show how this yields a strongly-polynomial running time. 
We see that by our improvement property, there are clearly $O(\log U)$ scaling
phases. We
accomplish this in Lemma \ref{lem:poly_phase}, by bounding the number of scaling phases at $O(m^{1/2})$. 

\begin{figure}[H]
\begin{framed}
\begin{alg}The final algorithm to find a max-flow on a bounded-degree network $G$.\end{alg}
\begin{pseudocode}{max-flow-1}{$G,\Delta$}
\item Construct the compact network $G_C$;
\item \textbf{for all} $u \in V_C$ \textbf{do} initialize $edge\text{-}list[u]$;
\item \textbf{while} there is a vertex $u \in V_C$ from which flow 
\item \stab can be discharged \textbf{do}
\item \tab Select $u:= \min\set{w \in V_A|d(w)}$
\item  \tab Iteratively call {\it discharge} on $u$ and
\item \tab\stab select the next vertex $u' := \min\set{w \in V_A|d(w)}$;
\item Run {\it restore-all-flows} and generate $d_f$ with {\it global-relabel};
\item $\Delta:= \min\set{\Delta/2,\max\set{2^{\ceil{\log e(u)}}|u \in V}}$; 
\end{pseudocode}
\end{framed}
\end{figure}

Correctness of our algorithm is immediate from the previous sections. We {\it initialize} the preflow $f$, and then 
iteratively call {\it max-flow-1} until $\Delta = 0$. We will now show that the algorithm terminates in $O(m^{1/2})$ phases.

\begin{lem}\label{lem:poly_phase} By iteratively calling max-flow-1, the algorithm terminates after $O(m^{1/2})$ scaling phases. \end{lem}
\begin{proof}
By Theorem \ref{thm:om_comp}, there are $O(m)$ vertices in the compact network across all scaling phases. Therefore, there are at least $O(m^{1/2})$ scaling phases where $c > m^{1/2}$. 

Now, we must place an upper bound on the number of scaling phases where $c \leq m^{1/2}$. Let $\Delta_i$ be the excess dominator for the $i$th scaling phase. We will now bound the number of  
\begin{inparaenum}[(1)]
\item arcs with $\Delta$-favorable and $\Delta$-large capacity, and
\item the number of active vertices. 
\end{inparaenum}
If we can show that there are $O(m^{1/2})$ arcs with $\Delta$-favorable and large capacity with respect
to $\Delta_i$, as well as $\Delta$-active vertices, we can conclude that both statements are valid for $O(m^{1/2})$ scaling phases.

(1) In an argument similar to Lemma \ref{lem:fav_in_comp}, we can show that any $\Delta$-favorable or $\Delta$-large arc will become $\Delta'$-abundant in 4 scaling phases.
Thus, we conclude that this statement is valid for $O(m^{1/2})$ scaling phases.

(2) Let $u \in V_A$ be an active vertex. By Lemma \ref{lem:act_in_comp}, we ``convert''  each active vertex to a vertex in $V_{SC}$ in the subsequent scaling phase. Therefore, the number
of vertices in $V_{SC}$ provides an upper bound on the number of vertices in $V_A$, since we can do this ``conversion'' in $1 = O(1)$ scaling phases.
\end{proof}

The following theorem is an immediate consequence of Theorem \ref{thm:nonsat_bound_all}, Theorem \ref{thm:dtree_prev}, and Lemma \ref{lem:poly_phase}.
\begin{thm}\label{thm:strong_rtime} By iteratively calling max-flow-1, the algorithm will find a maximum flow in $O(mn + m^{3/2}\log n)$ time. For bounded-degree networks, the algorithm
runs in $O(mn)$ time. By the preprocessing step in Appendix \ref{app:sparse}, the algorithm also finds a max-flow in $O(mn)$ time when $m = O(n)$.\end{thm}

\section{Conclusion}
We note that by resolving certain issues we would be able to formulate an $O(mn)$ time algorithm when $m = O(n^{2-\epsilon})$. Namely, a new scheme for counting low-capacity nonsaturating
pushes must be developed. We rely heavily on the fact that the degree of any vertex is $O(1)$, which allows us to limit low-capacity nonsaturating pushes to one per each value of $d_\ell$. We think that a 
technique described by Ahuja, et. al. \cite{ahuja89-2} in which such pushes are ``charged against'' high-capacity nonsaturating pushes could be a way to resolve this.

If this issue is resolved, we can solve the max-flow problem on general networks in $O(mn + m^{3/2}\log n)$ time, which implies an $O(mn)$-time algorithm for all but very dense networks. Further,
the elimination of the $O(\log n)$ factor would yield an $O(mn)$ time algorithm for all edge densities; due to limitations with dynamic trees (see \cite{pat06}), a new data structure would likely need
to be used.

\section*{Acknowledgments}
I would like to thank my advisor Janos Simon for his mentorship and guidance throughout this process. I would also like to thank  Geraldine Brady and Baker Franke for introducing me to theoretical computer science, as well as for their general support. Finally, I would like to thank Robert Tarjan for suggesting the technique that led to the extension of the algorithm to sparse networks.

\bibliographystyle{amsplain}
\bibliography{JGAA}

\providecommand{\bysame}{\leavevmode\hbox to3em{\hrulefill}\thinspace}
\providecommand{\MR}{\relax\ifhmode\unskip\space\fi MR }
\providecommand{\MRhref}[2]{%
  \href{http://www.ams.org/mathscinet-getitem?mr=#1}{#2}
}
\providecommand{\href}[2]{#2}
\begin{thebibliography}{10}

\bibitem{ahuja93}
R.~K. Ahuja, T.~L. Magnanti, and J.~B. Orlin, \emph{Network flows: theory,
  algorithms, and applications}, Prentice-Hall, Inc., Upper Saddle River, NJ,
  USA, 1993.

\bibitem{ahuja89}
R.~K. Ahuja and J.~B. Orlin, \emph{A fast and simple algorithm for the maximum
  flow problem}, Operations Research \textbf{37} (1989), no.~5, pp. 748--759
  (English).

\bibitem{ahuja89-2}
R.~K. Ahuja, J.~B. Orlin, and R.~E. Tarjan, \emph{Improved time bounds for the
  maximum flow problem}, SIAM J. Comput. \textbf{18} (1989), no.~5, 939--954.

\bibitem{cher90}
J.~Cheriyan, T.~Hagerup, and K.~Mehlhorn, \emph{Can a maximum flow be computed
  on o(nm) time?}, ICALP, 1990, pp.~235--248.

\bibitem{gold97}
B.~V. Cherkassky and A.~V. Goldberg, \emph{On implementing the push-relabel
  method for the maximum flow problem}, Algorithmica \textbf{19} (1997), no.~4,
  390--410.

\bibitem{christ11}
P.~Christiano, J.~A. Kelner, A.~Madry, D.~A. Spielman, and S.~Teng,
  \emph{Electrical flows, laplacian systems, and faster approximation of
  maximum flow in undirected graphs}, STOC, 2011, pp.~273--282.

\bibitem{ford56}
L.~R. Ford and D.~R. Fulkerson, \emph{{Maximal Flow through a Network.}},
  Canadian Journal of Mathematics \textbf{8} (1956), 399--404.

\bibitem{gold98}
A.~V. Goldberg and S~Rao, \emph{Beyond the flow decomposition barrier}, J. ACM
  \textbf{45} (1998), no.~5, 783--797.

\bibitem{gold88}
A.~V. Goldberg and R.~E. Tarjan, \emph{A new approach to the maximum-flow
  problem}, J. ACM \textbf{53} (1988), no.~4, 921--940.

\bibitem{gold97-2}
A.V. Goldberg and R.~Kennedy, \emph{Global price updates help}, SIAM J. Disc.
  Math \textbf{10} (1997), no.~4, 551--572.

\bibitem{kelner14}
J.~A. Kelner, Y.~T. Lee, L.~Orecchia, and A.~Sidford, \emph{An
  almost-linear-time algorithm for approximate max flow in undirected graphs,
  and its multicommodity generalizations}, SODA, 2014, pp.~217--226.

\bibitem{king92}
V.~King, S.~Rao, and R.~Tarjan, \emph{A faster deterministic maximum flow
  algorithm}, Proceedings of the Third Annual ACM-SIAM Symposium on Discrete
  Algorithms (Philadelphia, PA, USA), SODA '92, Society for Industrial and
  Applied Mathematics, 1992, pp.~157--164.

\bibitem{madry13}
A.~Madry, \emph{Navigating central path with electrical flows: From flows to
  matchings, and back}, FOCS, 2013, pp.~253--262.

\bibitem{orlin13}
J.~B. Orlin, \emph{Max flows in o(nm) time, or better}, Proceedings of the 45th
  Annual ACM Symposium on Theory of Computing (Palo Alto, CA, USA), STOC '13,
  ACM, 2013, pp.~765--774.

\bibitem{pat06}
M.~P\v{a}tra\c{s}cu and E.~D. Demaine, \emph{Logarithmic lower bounds in the
  cell-probe model}, SIAM J. Comput. \textbf{35} (2006), no.~4, 932--963.

\bibitem{sleat83}
D.~D. Sleator and R.~E. Tarjan, \emph{A data structure for dynamic trees}, J.
  Comput. Syst. Sci. \textbf{26} (1983), no.~3, 362--391.

\end{thebibliography}

\appendix

\section{Transforming Sparse Networks into Bounded-Degree Networks}\label{app:sparse}

The main idea centers around {\it converting} a network $G=(V,A)$ where $m = O(n)$ to a bounded-
degree network network $G'$ where $\max\set{\deg[u]|u\in V(G')} \leq k$, for some $k \in \mathbb{Z}_{> 0}$.
We can accomplish this by {\it splitting} vertices of ``relatively high'' degree into multiple vertices, and then
connecting original arcs in $A$ back to the new vertices. If we can bound the increase in $|V'|$ and $|A'|$,
then we simply can run this pre-processing on $G$, and then use $G'$ as an input into our original algorithm
for bounded-degree networks. We will pick the number $d := \floor{m/n}+3$ (the ``average degree'') as our goal;
our algorithm should produce $G'$ such that $d$ is the maximum in/out-degree. 

Unless otherwise specified, we call a graph where $\max\set{\deg[u]|u \in V} \leq k$ a $k$-bounded degree graph.

\subsection{The Preprocessing Algorithm}
\begin{figure}[H]
\begin{framed}
{
\begin{alg}An algorithm to convert a sparse network into a bounded-degree network.\end{alg}
\begin{compactenum}[]
\item \emph{Input:} A flow network $G$. 
\item \emph{Output:} A new flow network $G'$ such that $\forall u \in V(G')$, $\max\set{\deg[u]} \leq k$, for some $k \in \mathbb{Z}_{> 0}$. 
\item \begin{compactenum}[\bfseries Step 1.]
	\item Let $d:= \floor{m/n}+3$, $G' = \emptyset$.
	\item Let $V_d := \set{u \in V|\deg[u] > d}$.
	\item For all $u$ in $V_d$ do:
	\begin{compactenum}[\bfseries a.]
		\item Find $k_u$ such that $\deg[u] \leq k_u d$.
		\item Split $u$ into $k_u$ vertices $u_1 \ldots u_k$, and initialize a binary tree $T_u$ rooted at $u_1$.
		\item Connect all arcs originally incident to $u$ in $G$ to $u_1 \ldots u_k$ such that no degree exceeds $d$.
		\item Connect $u_2 \ldots u_k$ in the binary tree $T_u$ with undirected infinite capacity arcs.
	\end{compactenum}
	\item Return $\displaystyle G' = (V(G) \cup \bigcup_{u \in V_d} V(T_u), A(G) \cup \bigcup_{u \in V_d} A(T_u))$. 
\end{compactenum}
\end{compactenum}
}
\end{framed}
\end{figure}

\subsection{Running Time \& Correctness}
We will now state and prove several lemmas regarding the running time and correctness of the procedure. Moreover, we will show that $|V(G')|$ and $|A(G')|$ are within a constant factor of
$|V(G)|$ and $|A(G)$.

\begin{lem} $G'$ is a $d$-bounded degree network.\end{lem}
\begin{proof} Immediate from Step 3c.
\end{proof}
For sake of clarity in the following proofs, let $k = \max\set{k_u|u \in V_d}$. We also note that $k = O(1)$.
\begin{lem} $|V(G')| = O(|V(G)|) = O(n)$. \end{lem}
\begin{proof}
For each $u \in V_d$, we find $k_u$ such that $\deg[u] \leq k_u d \leq kd$ and add at most $k$ vertices to $G'$. Thus, the total number of vertices in $G'$ is
	\begin{align*} |V(G')| &= |V(G)| + \sum_{u \in V_d} k \\ &\leq |V(G)| + \sum_{i=1}^n k = \sum_{i=1}^n O(1) = O(n). \end{align*}
Thus $|V(G')| \leq O(n)$, which completes our proof.
\end{proof}

\begin{lem} $|A(G')| = O(|A(G)|) = O(n)$. \end{lem}
\begin{proof}
By Step 3d of the algorithm, we will add $k_u - 1 \leq k-1$ new arcs for the binary tree $T_u$ that is created. This will be done for each $u \in V_d$. Thus,
	\begin{align*} |A(G')| = |A(G)| + \sum_{u \in V_d} k - 1 &\leq |A(G)| + \sum_{i =1}^n k - 1 \\ &= |A(G)| + \sum_{i =1}^n O(1) = O(n) + |A(G)|. \end{align*}
Since $m = O(n)$, we have that $|A(G)| = O(n) \implies |A(G')| \leq O(n)$, which concludes our proof.
\end{proof}

The algorithm will be run {\it twice}; once to ensure that for all $u \in V$, the in-degree does not exceed $d$, and once more for
the out-degree.

\begin{lem}The preprocessing algorithm terminates in $O(n \log m/n)$ steps.\end{lem}
\begin{proof}
In the worst case, $|V_d| = n$ (i.e. every vertex has degree greater than $d$). Finding $k$ such that $\deg[u] \leq kd$ will take $\log d = \log m/n$ steps if we use repeated squaring.
Finally, connecting original arcs and creating new arcs in $T_u$ both take $O(k) = O(1)$ time. Thus, the overall cost of the pre-processing on $G$ will be $O(n \log m/n)$.
\end{proof}

\section{Manipulating Pseudoarcs with Dynamic Trees}\label{app:tc}
When manipulating $\Delta$-abundant and $\Delta$-small pseudoarcs, we use the dynamic trees data structure of Sleator and Tarjan \cite{sleat83}. In this section, we will list the operations supported by he data structure, and
detail several procedures for manipulating preflows in pseudoarcs.

\subsection{Tree Operations}
The following operations can be carried out in $O(k \log n)$ amortized time, over a sequence of $k$ tree operations. We will see that $k = m$, and all dynamic tree operations take at most $O(m \log n)$ time per scaling phase.
\begin{enumerate}[(1)]
	\item {\it initialize} - Creates an empty dynamic tree in $O(1)$ time.
	\item {\it link}$(u,v)$ -  Merges two trees containing $u$ and $v$ (assuming $root(u) \neq root(v)$). We let $\pi_u := v$. The root of the new tree is $root(v)$. Let $val(u) := r(v,\pi_v)$.
	\item {\it cut}$(u)$ - ``Breaks'' the tree containing $u$ by deleting the arc $(u,\pi_u) \in A_f$. $u$ becomes the root of the new tree, and $r(u,\pi_u) := val(u)$.
	\item {\it add-val}$(u,\delta)$ - Adds $\delta$ units of flow to each arc on $path(u)$.
	\item {\it find-min}$(u)$ - Returns $u$ such that $\min\set{val(u)|u \in path(u)}$.
\end{enumerate}

\subsection{Algorithm for Constructing Pseudoarcs}
The dynamic trees data structure is capable of efficiently carrying out a sequence of flow operations. Each tree is rooted at a vertex; for a tree containing vertex $u$, the root is denoted as $root(u)$. The parent of a vertex $u$
is denoted as $\pi_u$. Let $path(u)$ denote the simple, directed path $\ang{u,v_1,v_2,\ldots,v_k,root(u)}$. Let $val(u)$ denote the residual capacity of the arc $(u,\pi_u) \in A_f$. 

\begin{pseudocode}{feasible-path}{$u,\rho$}
\item \comment{A procedure for finding a feasible path in the input graph.}
\item $v := ${\it root}$(u)$;
\item \textbf{while} $v \not\in V_C \backslash \set{u}$ \textbf{do}
\item \tab Choose $(v,w)$ such that $r(v,w) > \rho$;
\item \tab {\it link}$(v,w)$;
\item \tab {\it enqueue}$(Q,\ang{\text{link},(v,w)})$
\item \tab $v := ${\it root}$(w)$; 
\item \textbf{return} feasible path $P$;
\end{pseudocode}

When we wish to create abundant pseudoarcs, we select $\rho = \Delta$. Similarly, when we wish to create $\Delta$-small pseudoarcs, we select $\rho = 0$.

\begin{pseudocode}{transfer-capacity}{$u,P$}
\item \comment{A procedure for transferring capacity of a path $P$ to a pseudoarc.}
\item Initialize a new pseudoarc $(u,v)$;
\item $\delta :=$ {\it find-min}$(u)$;
\item $v := ${\it root}$(u)$;
\item Add $(u,v)$ to $A_2$;
\item Decrease internal arc capacities by $\delta$;
\item Store path information and capacity in 
\item \stab the corresponding entry in $Q$;
\end{pseudocode}

\begin{pseudocode}{cut-all-saturated}{$u$}
\item \comment{A procedure to delete saturated arcs from the forest.}
\item $w :=$ {\it find-min}$(u)$;
\item $\delta := val(w)$;
\item \textbf{while} $\delta = 0$ \textbf{do}
\item \tab {\it cut}$(w)$
\item \tab {\it enqueue}$(Q,\ang{\text{cut},w})$;
\item \tab $w :=$ {\it find-min}$(u)$;
\item \tab $\delta := val(w)$;
\end{pseudocode}

Finally, we define the {\it create-all-pseudoarcs} procedure in order to create $\Delta$-abundant and $\Delta$-small
pseudoarcs.

\begin{pseudocode}{create-all-pseudoarcs}{$G_{in},\rho$}
\item \comment{A procedure to iteratively create all pseudoarcs.}
\item \textbf{while} there exists a feasible original vertex $u \in V_C$ \textbf{do}
\item \tab $u := \min\set{v \in V_C|d(v)}$;
\item \tab $p:=${\it feasible-path}$(u,\rho)$;
\item \tab {\it transfer-capacity}$(u)$;
\item \tab {\it cut-all-saturated}$(u)$;
\item \textbf{return} set $P$ of pseudoarcs;
\end{pseudocode}

As a result of {\it create-all-psedoarcs}, a pseudoarc constructed from $G_f$ could end up with capacity greater than $\Delta$ if there exist many parallel pseudoarcs between $u,w \in V_C$. In that case, we will
refer to it as $\Delta$-abundant. 

\subsection{Transforming the Compact Preflow}

All that is left to do now is transform the modified preflow on pseduoarcs in the compact network to the preflow in the residual network. We will accomplish this by sequentially recreating the forest of dynamic trees 
used to create the pseudoarcs. We will iteratively execute operations from $Q$ until there are none left. Doing so will allow us to recreate flows in $O(k \log n)$ time, over a sequence of $k$ operations.

\begin{pseudocode}{restore-all-flows}{$f$}
\item \comment{A procedure to transform flows in pseudoarcs to flows in residual arcs.}
\item Initialize a new residual preflow $f'$;
\item Create an empty dynamic tree;
\item \textbf{while} $Q \neq \emptyset$ \textbf{do}
\item \tab $\phi = ${\it dequeue}$(Q)$;
\item \tab Execute operation $\phi$ from the queue;
\item \tab \textbf{if} the min-path capacity $\delta > 0$ \textbf{then}
\item \tab \comment{ Update the residual preflow $f'$ }
\item \tab\tab {\it add-val}$(v,\delta)$; 
\item \tab\tab $f'[v,root(v)] := f'[v,root(v)] - \delta$;
\end{pseudocode}

We conclude this section with a theorem regarding the running time of the procedures described.

\begin{thm}\label{thm:dtree} Constructing abundant and small pseudoarcs takes $O(m \log n)$ time per scaling phase by calling {\it create-all-pseudoarcs.} It takes $O(m \log n)$ time to transform the compact preflow to the
residual preflow by calling {\it restore-all-flows}.\end{thm}
\begin{proof}
Each operation $\phi \in Q$ takes $O(\log n)$ amortized time over a sequence of $k$ operations. Since we can modify at most $O(m)$ arcs when constructing pseudoarcs, we see that $k = m$ and the overall
cost is $O(m \log n)$.
\end{proof}

\section{Efficiently Deleting Active Vertices}\label{app:del_act}

Finally, we must handle deletions of active vertices $u \in V_A$ when constructing abundant pseudoarcs. We can accomplish this in the abundance graph by preprocessing the graph and {\it splitting} each vertex $u \in V$ into
$u_{in}$ and $u_{out}$. All the in-arcs of $u$ will be connected to $u_{in}$, and similarly all out-acs to $u_{out}$. We then add an arc of infinite capacity between $u_{in}$ and $u_{out}$. When we wish to delete some active
vertex $u' \in V_A$, we can simply delete the arc $(u'_{in},u'_{out})$.

\end{document}